\documentclass[aip,jmp,numerical,12pt]{revtex4-1}
\usepackage{amssymb}
 \usepackage{amsthm,amsmath}
 \usepackage{hyphenat,hyperref}
\usepackage[utf8]{inputenc}
\usepackage[T1]{fontenc}
\usepackage{lineno}
\usepackage{xcolor, color, soul}
\usepackage{physics}
\newcommand{\stlines}{\mathbf{\tilde{T}}}
\def\C{\mathbb{C}} 
\def\uhp{{\mathbb H}}   
\def\dR{{\rm dR}}                    %The subindex dR standing for de Rham cohomology.
 \newtheorem{definition}{Definition}
 \newtheorem{theorem}{Theorem}
 \newcommand{\mat}[4]{
     \begin{pmatrix}
            #1 & #2 \\
            #3 & #4
       \end{pmatrix}
    }    
\begin{document}

% Use the \preprint command to place your local institutional report number 
% on the title page in preprint mode.
% Multiple \preprint commands are allowed.
%\preprint{}
%\begin{frontmatter}

%% Title, authors and addresses

\title{ Metric in the moduli of SU(2) monopoles from spectral curves and Gauss-Manin connection in disguise.}

%\title{} %Title of paper

% repeat the \author .. \affiliation  etc. as needed
% \email, \thanks, \homepage, \altaffiliation all apply to the current author.
% Explanatory text should go in the []'s, 
% actual e-mail address or url should go in the {}'s for \email and \homepage.
% Please use the appropriate macro for the type of information

% \affiliation command applies to all authors since the last \affiliation command. 
% The \affiliation command should follow the other information.

\author{Marcus A. C. Torres}
\email[]{mtorres@impa.br}
\homepage[]{www.impa.br/~mtorres}
%\thanks{}
%\altaffiliation{}
\affiliation{Instituto Nacional de Matemática Pura e Aplicada (IMPA),\\ Estrada Dona Castorina 110,22460-320, Rio de Janeiro,RJ, Brazil}

% Collaboration name, if desired (requires use of superscriptaddress option in \documentclass). 
% \noaffiliation is required (may also be used with the \author command).
%\collaboration{}
%\noaffiliation

\date{\today}

\begin{abstract}
% insert abstract here
We show here that from the metric of the manifold $M^0_2$, i.e., the reduced moduli of $SU(2)$ 2-monopoles in Yang-Mills-Higgs theory, one can recover the respective moduli of spectral curves using the method Gauss-Manin connection in disguise. This work is a step towards creating a inverse process of finding the metric of any $M^0_k$,  from spectral curves. This is a thirty years old problem that we hope to shed some light in it. %In this case, the method Gauss-Manin connection in disguise will show that the metric of $M^0_k$ can be written in terms of modular-type functions attached to the spectral curves. 
\end{abstract}

\pacs{}% insert suggested PACS numbers in braces on next line

\maketitle %\maketitle must follow title, authors, abstract and \pacs

% Body of paper goes here. Use proper sectioning commands. 
% References should be done using the \cite, \ref, and \label commands
\section{Introduction}
\label{intro}

The study of instantons and monopoles in three and four dimension are among the most comprehensive research areas at the interface of physics and mathematics. From the physics point of view, they relate to solitons, dualities and non-perturbative  Yang-Mills theories. From the mathematical point of view, it involves knowledge of Analysis, Differential Geometry, Algebraic Geometry and Twistor theory. In this work, we add elements of Hodge theory to relate the metric of the moduli space of charge k monopoles $M^0_k$  and its spectral curves in SU(2) Yang-Mills-Higgs (YMH) theory in three spacial dimensions (static monopoles).

YMH monopoles in three dimensions are equivalent to instanton solutions of Yang-Mills theory in Euclidian four dimensions constrained by the fact that gauge fields do not depend on the fourth direction. In this equivalence, the Higgs field is the fourth component of the gauge field in four dimensions. In this way, the twistor methods applied in instantons were also adapted to YMH monopoles \cite{1982CMaPh..83..579H}.

In a recent work \cite{Morales:2017phm}, while revising the  metric of $M^0_2$ reduced moduli of 2-monopoles, the  present author noticed that  the moduli of enhanced elliptic curves obtained from the self-dual metric  should correspond to the moduli of spectral curves of 2-monopoles in its enhanced version. Notice\cite{1982CMaPh..83..579H} that the spectral curve of a $k$-monopole is an algebraic curve of genus $(k-1)^2$  and for $k$=2, the spectral curve is an elliptic curve \cite{Hurtubise:1983hl}.

The  moduli of enhanced elliptic curves appears from the metric of $M^0_2$ using the method Gauss-Manin connection in disguise developed by Movasati\cite{alim2016gauss,ho06-1,ho14} which shows that  the set of Darboux-Halphen differential equations obtained from self-duality of the metric of $M^0_2$ are a vector field in the moduli of an enhanced elliptic curve. 

In here, the guess made in our previous collaboration\cite{Morales:2017phm} is proved. 

In sections \ref{kmonopoles} and \ref{scurves} we review basic elements of  SU(2) monopoles and spectral curves following closely some original articles  \cite{atiyah1977instantons,1982CMaPh..83..579H,Hurtubise:1983hl} and reference books \cite{atiyah2014geometry, ward1990twistor}.  
In section \ref{gmcd} we quickly review the program Gauss-Manin connection in disguise for elliptic curves\cite{ho14}  and in section \ref{spectraltometric}  we show how the moduli of spectral curves of 2-monopoles emerge from the metric of $M^0_2$. In \ref{conclusion}  we summarize this article and comment about the cases $k>2$ and the issue of different parametrization of universal families of curves and the  weights of the respective set of modular-type functions attached to such curves.

\section{\textit{k}-Monopoles}\label{kmonopoles}
A $k$-monopole or BPS-monopole of charge $k$ in Yang-Mills-Higgs theory is a static soliton in  $\mathbb{R}^3$ that is a solution of the Bogomolny equation \cite{Bogomol`nyi_1976}:
\begin{align}
F=\star D \phi, \qquad \textnormal{with}\label{Bogomolny}\\
F:=\mathrm{d}{A} + {A}\wedge{A}\quad \textrm{and}\quad D:= \mathrm{d}+{A},\nonumber
\end{align}
where ${A}$ is the gauge field or connection form on a principal SU(2)-bundle over $\mathbb{R}^3$, $F$ is its curvature 2-form or field strength, $D$ is the covariant exterior  derivative or connection and the Higgs field $\phi$ is a section of the associated  $\mathfrak{su}(2)$-bundle. $\star$ is the Hodge dual operation and the Bogomolny equation (\ref{Bogomolny}) is part of the self-duality equations for the related instantons in four dimensions.

The monopole solution has also to satisfy the finite action condition $\int |F|^2<\infty$ and the boundary condition
\begin{equation}
	|\phi|=1-\frac{k}{2r}+ O(r^{-2})\quad \text{as}\;\;r\rightarrow \infty,
\end{equation}  
where the charge $k$ is an integer number.

A clever treatment of monopoles was given by Hitchin \cite{1982CMaPh..83..579H} where he applied the twistor methods in the space $\stlines$ of oriented straight lines (geodesics) in $\mathbb{R}^3$. $\stlines$ has a holomorphic structure given by cross product and $\stlines\equiv\mathbf{T}\mathbb{P}_1(\mathbb{C})$ the holomorphic tangent bundle of the projective line.

Then   the solutions of Bogomolny equations were restated  in terms of complex geometry of $\stlines$ where  the spectral curves were introduced \cite{1982CMaPh..83..579H,hitchin1983construction}.

First consider $E$ the  rank 2 complex vector bundle on $\mathbb{R}^3$ associated to the principal SU(2) bundle. Now one defines a rank 2 complex vector bundle $\tilde{E}$ by defining at each point $z\in \stlines$ a fiber  $E_z$. To each point $z\in \stlines$ there is a corresponding oriented line $l_z \in \mathbb{R}^3$. $E_z$ is given by the the space of sections $s$ of E with support on $l_z$   such that 
\begin{equation}
(u^jD_j-i\phi)s=0,\label{holomrank2}
\end{equation} 
$u$ is the unit tangent vector pointi ng (in the positive direction) along the oriented line $l_z$. It follows from Bogomolny equations that $\tilde{E}$ has a natural holomorphic structure \cite{1982CMaPh..83..579H}. Conversely, from a holomorphic vector bundle $\tilde{E}$ on $\stlines$ one  reconstructs the solution $(A,\phi)$ to Bogomolny equations. But not all section $s$ of $\tilde{E}_l$ satisfy the boundary conditions, which is the vanishing of $s$ at both ends of $l$. %Following \cite[Chapter~9]{atiyah2014geometry}

\section{Spectral curve of a k-monopole}\label{scurves}

For each oriented line $l\in \mathbb{R}^3$, the space of solutions (\ref{holomrank2}) which decay at $+\infty$ is one-dimensional. This space is a holomorphic line bundle and a subbundle of $\tilde{E}_l$ and it belongs to a class of ans\"atze $\mathcal{A}_k$ according to the charge $k$ of the monopole \cite{atiyah1977instantons}. Furthermore, the set of lines for which equation (\ref{holomrank2}) has a solution decaying to zero at both ends forms a compact algebraic curve $S$ in $\mathbf{T}\mathbb{P}_1(\mathbb{C})$. $S$ is called the spectral curve and it has genus $(k-1)^2$.

Inhomogeneous coordinates $(\eta,\zeta) $on $\mathbb{P}_1(\mathbb{C})$ gives local coordinates on $\stlines: (\eta,\zeta) \rightarrow \eta \partial/\partial \zeta$. 
$S$  in terms of such local coordinates  is given by
\begin{equation}
p(\eta,\zeta)=\eta^k+a_1(\zeta)\eta^{k-1}+\dots a_k(\zeta)=0,\label{Spolinomial}
\end{equation}
where $a_i(\zeta)$ is a polynomial of degree 2i.

The polynomial $p(\eta,\zeta)$ is preserved by an antiholomorphic involution $\tau(\eta,\zeta)= (-\overline{\eta}/\overline{\zeta}^2,-\overline{\zeta}^{-1})$, a real structure on $\mathbf{T}\mathbb{P}_1(\mathbb{C})$. Therefore $p(\eta,\zeta)$ depends on $(k+1)^2-1$ real parameters. Since $S$ is constrained by its genus (transcendental or ES  constraint\citep{Ercolani:1989tp}),   the parameter space has dimension 
\begin{equation}
(k+1)^2-1-(k-1)^2= 4k-1.
\end{equation}
This is the dimension of the moduli space of k-monopoles $M_k$. Out of these parameters, the center of mass position of a k-monopole in $\mathbb{R}^3$ can be translated to the origin and the remaining parameter space corresponds to the reduced moduli space $M^0_k$ with dimension $4k-4$.
%define a rank 2 vector bundle $\tilde{E}$ on $\stlines$
% \newtheorem{E}{Theorem} 
% \begin{E}
% If $(D,\phi)$ satisfy the SU(2) Bogomolny equations $D\phi=\star F$, then $\tilde{E}$ is in a natural way a holomorphic vector bundle on the space of geodesics $\stlines$ such that 
% \begin{enumerate}[(i)]
% \item $\tilde{E}$ is trivial on every real section.
% \item $\tilde{E}$ has a sympletic structure
% \end{enumerate}
% \end{E}

\subsection{Spectral curve for {\rm k}= 2}\label{scurve2}
This case was extensively studied by Hurtubise \cite{Hurtubise:1983hl}. The spectral curve $S$ is an elliptic curve of genus 1. The real structure of $S$ imposes via Weierstrass $\mathfrak{p}$-function the complex structure $\tau$ of the corresponding torus $\mathbb{C}/\Lambda$ to be purely imaginary and its  corresponding lattice $\Lambda$ to be rectangular.

Factoring out six parameters from translation action and $SO(3)$ action from the polynomial (\ref{Spolinomial}) for the spectral curve of 2-monopoles, two real parameters remain:

\begin{equation}
\eta^2= r_1\zeta^3- r_2\zeta^2-r_1\zeta, \;r_i\in \mathbb{R}, r_1\ge 0 \label{S2polynomial}
\end{equation}
The genus constraint is enforced by matching the above equation to a normal form of the elliptic curve.

First notice that when $r_1=0$, the spectral curve degenerates to two $k=1$ spectral curves.
\begin{equation}\label{scurve1}
\eta=i\frac{\pi}{2}\zeta\quad \textnormal{and} \quad \eta=-i\frac{\pi}{2}\zeta.
\end{equation}
 In this case $r_2=\pi^2/4$, as we show below (\ref{r2}), and there is no free parameter. This is the case\cite{Ward:1981jb,Sutcliffe:1997ec} where a 2-monopole is simply a superposition of two 1-monopoles both centered at the origin of $\mathbb{R}^3$  and it agrees with the fact that the  dimension of the reduced moduli $M^0_1$ is zero. The spectral curves (\ref{scurve1}) are two complex lines tangent  to two different points in $\mathbb{P}_1\equiv S^2$. The symmetries of the spectral curve determines a symmetry of the monopole. In this case, the isotropy group $S^1\times \mathbb{Z}_2$ of the two $k=1$ spectral curves corresponds to the axial symmetry of the two 1-monopoles solution and the exchanging the two 1-monopoles. 

For $r_1>0$, the spectral curve can be reparametrized to:
\begin{gather}\label{g2g3}
\tilde{\eta}^2= 4\tilde{\zeta}^3-g_2(\Lambda)\tilde{\zeta}-g_3(\Lambda), \;\textrm{where},\\ \nonumber 
\tilde{\eta}=\eta (4/r_1)^{1/2},\quad\tilde{\zeta}=\zeta -\tfrac{r_2}{3r_1}, \\\nonumber
g_2(\Lambda)=60 G_4(\Lambda)= 12(r_2/3r_1)^2+4\quad\textrm{and},\\
g_3(\Lambda)=140 G_6(\Lambda)=8(r_2/3r_1)^3+4(r_2/3r_1).\label{eisenstein}
\end{gather}
and $G_4$ and $G_6$ are Eisenstein series of weight 4 and 6, respectively, functions of the retangular lattice $\Lambda$ with real generator $l_r=\sqrt{4r_1}$ and imaginary generator $l_i$.

A homothetic scaling of the lattice transform $g_2$ and $g_3$:  
\begin{equation}\label{gscaling}
g_i(m\Lambda)=m^{-2i}g_i(\Lambda),\, i=2,3,\quad m\in \mathbb{R}^*
\end{equation}
and the polynomial (\ref{g2g3}) is preserved if we reparametrize $(\tilde{\eta},\tilde{\zeta})$ to absorb such scaling
\begin{equation}\label{polyscaling}
\tilde{\eta}\rightarrow m^{-3}\tilde{\eta}\quad\textnormal{and}\quad \tilde{\zeta}\rightarrow m^{-2}\tilde{\zeta}.
\end{equation}
Therefore we should consider modular functions such as $I= 27g_3^2/g_2^3$.  This function will be  invariant to scaling of the lattice $\Lambda$ and it will only depend on the ratio of the generators $\tau=l_i/l_r$, a purely imaginary number.  To make this dependency explicit we can reparametrize the variables as above  with $m= l_r^{-1}$ and obtain $g_2(\tau)= 16r_1^2g_2(\Lambda)=60G_4(\tau)$ and $g_3(\tau)= (4r_1)^3g_3(\Lambda)=140G_6(\tau)$. In terms of the normalized Eisenstein series $E_{2i}$ and Riemann zeta functions, $G_{2i}(\tau)=2\zeta(2i)E_{2i}(\tau)$ with,
\begin{equation}
\label{eisenstein}
E_{2i}(q):=1+b_i\sum_{n=1}^\infty \left (\sum_{d\mid n}d^{2i-1}\right )q^{n},\ \ i=1,2,3,
\end{equation}
with $(b_1,b_2,b_3)=(-24, 240, -504)$  and $q= e^{i2\pi\tau},\;\tau\in \mathbb{H}=\{\tau \in \mathbb{C}|\Im(\tau)\ge 0)$ for convergency of the series. When $\tau$ is purely imaginary, $E_{2i}(\tau)$ take values in $\mathbb{R}$.

From (\ref{g2g3}), $I$ depends on the ratio $(r_1/r_2)^2$. Notice that in the limit $r_1\rightarrow 0$, the discriminant of the elliptic curve $\Delta= g_2^3-27g_3^2=0$ and  $I(\tau)=1$. This corresponds to the limit $\tau\rightarrow i \infty$. In order to proceed showing that $r_2=\pi^2/4$ in this limit, we explore $I$ near 1.

\begin{gather}\label{Inear1}
\frac{1-I(\tau)}{27}= \frac{64}{j(\tau)}= 2^{12}3^3(q-744q^2+\dots)\quad\textnormal{where,} \\\nonumber
j(\tau)=\frac{1728g_2^3}{\Delta}\;\textnormal{is the Klein modular function }.
\end{gather}
From (\ref{eisenstein}), 
\begin{equation}\label{Ir}
\frac{1-I}{27}=\frac{r^4(\frac{1}{4}+r^2)}{(1+3r^2)^3},\quad\textnormal{with }  r=\frac{r_1}{r_2}.
\end{equation}
We see that the limit $r_1,r\rightarrow 0$ coincides with $\tau\rightarrow i\infty $ or $ q\rightarrow 0$. Near this limit, we keep only the first term of the Eisenstein series $G_4(\tau)$. From (\ref{eisenstein}):
\begin{gather}\label{r2}
g_2(\Lambda)= \frac{1}{16 r_1^2}\frac{64}{3}r_2^2(1+3r^2)=\frac{1}{16 r_1^2}60G_4(\tau)\xrightarrow{\:q \to 0 \: }\frac{1}{16 r_1^2} \frac{(2\pi)^4}{12}.\\
\textnormal{Therefore,}\quad r_2\xrightarrow{\:q \to 0 \: } \pi^2/4.
\end{gather}
 Hence, the point $(r_1,r_2)=(0,\pi^2/4)$ corresponds to the singular point $(\Delta=0)$ of the real elliptic curve $S$ factored out by $SO(3)$ action and $\mathbb{R}^3$ translations, corresponding to $\tau=i\infty$. The total space of parameters has one real dimension and it corresponds to purely imaginary $\tau, \, -i\tau\in \mathbb{R}_{\geq 0}\cup \infty $.

\section{Gauss Manin connection in disguise}\label{gmcd}
In \cite{movasati2008differential,movasati2012ramanujan} Movasati realized that the Ramanujan relations between Eisenstein series can be computed using the Gauss-Manin connection of families of elliptic curves. Later in a private communication, Pierre Deligne called "Gauss-Manin connection in disguise" the vector field that best express the property of Griffths transversality\cite{griffiths1968periods1,griffiths1968periods2} of a  Gauss-Manin connection. Since then the method Gauss-Manin connection in disguise has been applied in many families of algebraic curves and relating them to differential equations and automorphic forms or modular-type functions \cite{alim2016gauss,movasati2014gauss,movasati2015gauss,movasati2016gauss}. 

Our interest are in finding differential equations in the universal families of spectral curves of $k$-monopoles. The method developed for the elliptic curve \cite{movasati2012ramanujan,ho14} still need to be thought through for spectral curves because of the reality condition on the spectral curves, but our general argument is that the reality condition is lifted for the sake of finding the Gauss-Manin connection and the respective vector field and later the reality condition is imposed on the domain of solutions of the vector field equations.  

We present here the two known cases of  families of enhanced elliptic curves which correspond to  geometric expressions of Ramanujan and Darboux-Halphen differential equations. A good review is in Movasati's lectures\cite{ho14}. In both cases, the idea is to define the moduli of enhanced elliptic curve  by including information about its Hodge structure. Then, one calculates its Gauss-Manin connection and finds the appropriate vector field.  
 
\subsection{Ramanujan differential equations}
We extend the one-parameter family of elliptic curves (\ref{g2g3}). Recall that the first de Rham cohomology $H^1_\dR(E)$ of an elliptic curve $E$ is a two-dimensional vector space. The moduli $\mathtt{T_R}$ of pairs $(E,[\alpha, \omega])$, where    $\alpha, \omega\in H^1_\dR(E)$ are a basis of the cohomology classes  of  1-forms in $E$ with $\alpha$  a regular differential 1-form on $E$ and $\omega$ such that $\langle\alpha,\omega\rangle=1$. In a equivalent way, $\mathtt{T_R}$ can be defined as the moduli of pairs$(E,[\omega])$, $\omega\in H^1_\dR(E)\backslash F^1$ and there is a unique regular 1-form $\alpha$ in the Hodge filtration $ F^1\subset H^1_{\dR}(E)$ such that $\langle\alpha,\omega\rangle=1$. Therefore $\mathtt{T_R}$ is a three-dimensional space and it has a  corresponding  universal family of elliptic curves
\begin{equation}\label{Et}
E_t: y^2=4(x-t_1)^3-t_2(x-t_1)-t_3,
\end{equation}
with $\alpha = [\frac{dx}{y}],\; \omega=[\frac{xdx}{y}]$ and the moduli $\mathtt{T_R}$ can be expressed as
\begin{equation}
\mathtt{T_R}:=\{(t_1,t_2,t_3)\in \C^3| \Delta=t_2^3-27t_3^2\neq 0\}.
\end{equation}
The Gauss-Manin connection of the above family $E_t$, written in the basis $(\alpha,\omega)$ is given by
\begin{equation}
\label{4maR}
 \nabla\begin{pmatrix}\alpha\\ \omega\end{pmatrix}=
A
\begin{pmatrix}\alpha\\ \omega\end{pmatrix},
\end{equation}
where
\begin{gather}
A=\frac{1}{\Delta}
\begin{pmatrix}
-\tfrac{3}{2}t_1\beta-\tfrac{1}{12}d\Delta&\tfrac{3}{2}\beta\\ 
\Delta dt_1-\tfrac{1}{6}t_1d-(\tfrac{3}{2}t_1^2+\tfrac{1}{8}t_2)\beta&\tfrac{3}{2}t_1\beta\Delta+\tfrac{1}{12}d\Delta
\end{pmatrix},\\ \nonumber
\beta= 3t_3dt_2-2t_2dt_3.
\end{gather}

 In $\mathtt{T_R}$ there is a unique vector field $R$ such that\cite{ho14}   
\begin{equation}\label{GMCD}
\nabla_R(\alpha)=-\omega,\quad \nabla_R(\omega)= 0.
\end{equation}

The vector field $R$ is given by the Ramanujan differential equations\cite{ra16}
\begin{gather}\label{Ramanujaneq}
\left\lbrace
\begin{aligned}
\frac{\partial t_1}{\partial \tau}&=t_1^2-\tfrac{1}{12}t_2,\\
\frac{\partial t_2}{\partial \tau}&=4t_1t_2-6t_3,\\
\frac{\partial t_3}{\partial \tau}&=6t_1t_3-\tfrac{1}{3}t_2^2,
\end{aligned}
\right.
\end{gather}
where $\tau$ is a direction in the moduli $\mathtt{T_R}$ chosen by $R$.

$R$ has been called Ramanujan vector field, modular vector field and lately, Gauss-Manin connection in disguise. 

It may seem that in this process the moduli was not only enhanced but also enlarged since the moduli of elliptic curves has 1 complex dimension while $\dim(\mathtt{T_R})=3$. But if we look to the solution of (\ref{Ramanujaneq}),
\begin{equation}\label{Ramanujansolution}
(t_1(\tau),t_2(\tau),t_3(\tau)):=(\frac{2\pi i}{12}E_2(\tau),\ 12(\frac{2\pi i}{12})^2E_4(\tau), 8(\frac{2\pi i}{12})^3E_6(\tau)),
\end{equation}
 
$E_4$ and $E_6$ are modular forms of weight $k=4, 6$, respectively,
\begin{equation*}E_k(\dfrac{a\tau+b}{c\tau+d})=({c\tau+d})^kE_k(\tau),\; \begin{pmatrix} a&b\\c&d\end{pmatrix}\in SL(2,\mathbb{Z}), 
\end{equation*}
and $E_2$ is a quasi-modular form of weight 2:
\begin{equation*}E_2(\dfrac{a\tau+b}{c\tau+d})=({c\tau+d})^2E_2(\tau)+ \frac{12}{2\pi\mathrm{i}}c(c\tau+d) .
\end{equation*}

Therefore, $R$ vector field corresponds to a map of $\tau\in \mathbb{H}$ to $ (t_1,t_2,t_3)\in \mathtt{T_R}$. 

The transformation of $E_2, E_4,$ and $E_6$, under the action of the modular group $SL(2,\mathbb{Z})$ on $\tau$, preserving the lattice that defines the elliptic curve $E_t\equiv \C^2/\Lambda_{\tau}$, reveals that   
 there is a group of isomorphisms ${G}$ that acts on $\mathtt{T_R}$.  The quotient moduli $\mathtt{T_R}/{G}$ has one complex dimension.  
For $g\in{G}$, 
\begin{eqnarray*} \label{G}
[\alpha,\omega]&\xrightarrow{\bullet g}&[\alpha,\omega]g=[c\alpha,c'\alpha+\omega /c]\\
t=(t_1,t_2,t_3)&\xrightarrow{\bullet g}&t'=(c^{-2}t_1+c'/c,c^{-4}t_2,c^{-6}t_3)\\
(x,y)&\xrightarrow{\bullet g}&(c^{-2}x+c'/c,c^{-3}y)\\
{E_t}&\equiv&{E_{t'}} %({E_t}: \;y^2= 4(x-{t_1})^3-{t_2}(x-{t_1})-{t_3})
\end{eqnarray*}
$$g=\begin{pmatrix}c&c'\\0&c^{-1}
\end{pmatrix}=\begin{pmatrix}c&0\\0&c^{-1}
\end{pmatrix}\begin{pmatrix}1&c'/c\\0&1
\end{pmatrix},\quad c\in \C^*, c'\in \C$$
Notice that this group action preserves the intersection form $\langle\alpha,\omega\rangle=\langle c\alpha,c'\alpha+\omega/c\rangle=1$.

\subsection{Darboux-Halphen differential equations}
\label{hddeqs}

In this case, the enhanced elliptic curve is given by a triple  $(E,(P,Q),\omega)$, where $E$ is an elliptic curve, $\omega\in H^1_\dR(E)\backslash F^1$,
and $P$ and $Q$ are  a pair of points of $E$ that generates the $2$-torsion subgroup 
 with the Weil pairing $e(P,Q)=-1$. 
The points $P$ and $Q$ are given by $(T_1,0)$ and $(T_2,0)$. In here, the torsion data is necessary because the modular group , or group of lattice equivalence, of this enhanced curve is the congruence subgroup $\Gamma(2)\subset SL_2(\mathbb{Z})$, which has index $\left[SL_2(\mathbb{Z}):\Gamma(2)\right]=6$. The torsion data choose one out of six enhanced elliptic curves with same $(E,\omega)$ pairs.

For each choice of $\omega$, there is a unique regular differential 1-form in the Hodge filtration $ \omega_1\in F^1$, such that $\langle \omega,\omega_1\rangle=1$ and $\omega,\, \omega_1$ together form a basis of $H^1_\dR(E)$. The corresponding universal family of elliptic curves is given by

\begin{align}\label{ET}
E_T:\ \ y^2-4(x-T_1)(x-T_2)(x-T_3)=0,\\\nonumber 
\textnormal{and moduli } \mathtt{T_H}=\{(T_1,T_2,T_3)\in \C^3|\, T_1\neq T_2\neq T_3\}.
\end{align}

 In fact, this universal family patches together all six enhanced elliptic curves, separated by singularity borders $T_i=T_j$, due its symmetry under permutation of $T_1, T_2$ and $T_3$. Hence, it is a six-fold cover of the enhanced elliptic curve $(E,\omega)$ that is isomorphic to the enhanced elliptic curve $(E,[\alpha,\omega])$ for the full modular group $SL_2(\mathbb{Z})$. 

The Gauss-Manin connection of the family of elliptic curves $E_T$ written in the basis $\frac{dx}{y},\ \frac{xdx}{y}$ is given
as bellow:
\begin{equation}
\label{4maH}
 \nabla\begin{pmatrix}\frac{dx}{y}\\ \frac{xdx}{y}\end{pmatrix}=
A
\begin{pmatrix}\frac{dx}{y}\\ \frac{xdx}{y}\end{pmatrix}
\end{equation}
where 
\begin{align*}
A &=& 
 \frac{dT_1}{2(T_1-T_2)(T_1-T_3)}\mat{-T_1}{1}{T_2T_3-T_1(T_2+T_3)}{T_1}+\\
 & &
\frac{dT_2}{2(T_2-T_1)(T_2-T_3)}\mat{-T_2}{1}{T_1T_3-T_2(T_1+T_3)}{T_2}+\\
& & \frac{dT_3}{2(T_3-T_1)(T_3-T_2)}\mat{-T_3}{1}{T_1T_2-T_3(T_1+T_2)}{T_3}.
\end{align*}
In the parameter space of the family of elliptic curves $E_T$ there is a unique vector field $H$, such that
\begin{equation}
\label{niteroi}
\nabla_{H}(\frac{dx}{y})= -\frac{xdx}{y},\ \nabla_{H}(\frac{xdx}{y})= 0.
\end{equation}
The vector field $H$ is given by the Darboux-Halphen differential equation
\begin{gather}\rm
\label{halphen}
\left \{ 
\begin{aligned}
\frac{\partial T_1}{\partial\tau}=T_1(T_2+T_3)-T_2T_3,\\ 
\frac{\partial T_2}{\partial\tau}= T_2(T_1+T_3)-T_1T_3, \\
 \frac{\partial T_3}{\partial\tau}= T_3(T_1+T_2)-T_1T_2.
\end{aligned} 
\right.
\end{gather}
where $\tau$ is a direction in $\mathtt{T_H}$ chosen by $H$. This vector field $H$ has been called Darboux-Halphen vector field, and lately, Gauss-Manin connection in disguise.
 Similarly to the previous subsection, there is a group of isomorphism $G'$ in $\mathtt{T_H}$ with two generators ( addictive and multiplicative) and the quotient $\mathtt{T_H}/G'$ has one complex dimension. For $g\in{G'}$,
\begin{eqnarray*}
T=(T_1,T_2,T_3)&\xrightarrow{\bullet g}&T'=(c^{-2}T_1+c'/c,c^{-2}T_2+c'/c,c^{-2}T_3+c'/c),\\
(x,y)&\xrightarrow{\bullet g}&(c^{-2}x+c'/c,c^{-3}y),\\
{E_T}&\equiv&{E_{T'}}, \qquad  c\in \C^*,\; c'\in \C %({E_t}: \;y^2= 4(x-{t_1})^3-{t_2}(x-{t_1})-{t_3})
\end{eqnarray*}
 
\subsection{\texorpdfstring{$\mathtt{T_R}$}{Lg} and \texorpdfstring{$\mathtt{T_H}$}{Lg}}\label{TRTH}
There is a algebraic morphism between the moduli $f: \mathtt{T_H}\longrightarrow  \mathtt{T_R}$  given by a match between the elliptic curves (\ref{4maR}) and (\ref{4maH}):
\begin{gather}\label{morphismHR}
(T_1,T_2,T_3)\rightarrow(T, -4\sum_{1\leq i< j\leq 3}(T-T_i)(T-T_j),4(T-T_1)(T-T_2)(T-T_3)),\\\nonumber
\textnormal{where }\quad T=\frac{T_1+T_2+T_3}{3} 
\end{gather}
Since the permutations of $T_1, T_2$ and $T_3$ in $\mathtt{T_H}$ are mapped to the same point in $\mathtt{T_R}$, $f$ is a six to one map, but if we restrict to the region $|T_1|<|T_2|<|T_3|$ in $\mathtt{T_H}$, $f$ is an isomorphism. 

\section{From the metric of the moduli space of 2-monopoles to the spectral curve}\label{spectraltometric}
%Most of the results obtained here are in reference of the $k=2$ case. 
%\subsection{Reviewing the case k=2}

In their book\cite{atiyah2014geometry}, Atiyah and Hitchin showed that the reduced moduli $M^0_2$ of 2-monopoles  is a four dimensional hyperk\"ahler manifold and an anti-self-dual  Einstein manifold. Since $M^0_2$ admits $SO(3)$ isometry, the metric is a Bianch IX \cite{Bianchi2001}. This is consequence of  the hyperkähler structure of $M^0_2$ which has an $S^2$-parameter family of complex structures: if I, J, K are covariant constant complex structures in $M^0_2$ then $aI+bJ+cK$ is also a covariant constant complex structure in $M^0_2$ given that $a^2+b^2+c^2=1,\, (a,b,c)\in \mathbb{R}^3$. The SO(3) isometry rotates this $S^2$ in a standard way. Following Atiyah and Hitchin\cite[Chapter~8,9]{atiyah2014geometry} and our review\cite{Morales:2017phm}, the 4-dimensional Bianchi IX metric is cast in the form
\begin{equation}
\label{bianchiix} ds^2 = (abc)^2 d\rho^2 + a^2  (\sigma_1)^2 + b^2(\sigma_2)^2 + c^2  (\sigma_3)^2,
\end{equation}
 where $a, b, c$ are real functions of $\rho$ which  parametrizes  each SO(3) orbit in $M^0_2$ or, in other words,  it parametrizes  trajectories orthogonal to these orbits in $M^0_2$. The $SO(3)$ invariant 1-forms $\sigma_i$ are dual to the standard basis $X_1, X_2, X_3$ of its Lie algebra. They obey the structure equation 
 \begin{equation}
 d\sigma_i=-\sigma_j\wedge\sigma_k,
 \end{equation}
for all cyclic permutations  $(i, j, k)$ of $(1, 2, 3)$.% $\eta$ measures the distance in a trajectory orthogonal to the $SO(3)$ orbits, up to a conformal factor.

The self-duality equations lead to the following equation
\begin{equation}\label{sdeq}
\frac{2}{a}\frac{da}{d\rho}= b^2+c^2-a^2- 2bc,
\end{equation}
and two other equations obtained by cyclic permutations  of (a,b,c). Upon reparametrization 
\begin{equation}\label{reparam}
a^2=\frac{\Omega_2\Omega_3}{\Omega_1},\quad b^2=\frac{\Omega_3\Omega_1}{\Omega_2},\quad c^2=\frac{\Omega_1\Omega_2}{\Omega_3}.
\end{equation}
we obtain from (\ref{sdeq}) the three Darboux-Halphen differential equations:
\begin{equation}\label{halphereal}
\dot{\Omega}_i= \Omega_i(\Omega_j+\Omega_k)-\Omega_j\Omega_k,
\end{equation}
where (i,j,k) run over  cyclic permutations of (1,2,3) and the derivative (denoted by dot) is with respect to $\rho$, a real parameter. When we put together the fact that the $k=2$ spectral curve corresponds to an elliptic curve with purely imaginary $\tau$ and consequently, real valued Eisenstein series $E_{2i}(\tau)$, we conclude that the solution (\ref{Ramanujansolution}) of Ramanujan equations (\ref{Ramanujaneq}) $(t_1(\tau),t_2(\tau),t_3(\tau))$ with purely imaginary  $t_1$ and $t_3$  will only match $(\Omega_1,\Omega_2,\Omega_3)\in \mathbb{R}^3$ via $f$ morphism (\ref{morphismHR}) if, 
\begin{equation}\label{tauomega}
\tau=i\rho\quad \textnormal{and}\quad \Omega_j=iT_j
\end{equation} 
From the  discussion in \ref{hddeqs}, the space 
\begin{equation}
\mathtt{T_\Omega}:=\{(\Omega_1,\Omega_2,\Omega_3)\in \mathbb{R}^3| \Omega_1\neq\Omega_2\neq\Omega_3\},
\end{equation}
 corresponds to a section of purely imaginary coordinates $(T_1,T_2,T_3)$ of the moduli space $\mathtt{T_H}$ of the enhanced elliptic curve $E_T$. As shown below, this sectioning corresponds to the reality condition of the spectral curve.  In $\stlines$, the reality condition is given by a real structure $\varsigma$ that acts by reverting the orientation of the corresponding lines in $\mathbb{R}^3$ \cite{1982CMaPh..83..579H}. In terms of $\stlines$ coordinates, it corresponds to $\varsigma(\eta,\zeta)= (-\overline{\eta}/\overline{\zeta}^2,-1/\overline{\zeta})$. Hence, a real curve in $\stlines$ is invariant under $\varsigma$ as in (\ref{S2polynomial}). For the reparameterized  spectral curve (\ref{g2g3}) this means that the coefficients $g_2, g_3$ of the defining polynomial are real. Looking at their values in terms of Eisenstein series and using the group $G$ of isomorphisms (\ref{G}) of the moduli $\mathtt{T_R}$ with $c= (-2\pi i)^{1/2}$, we see that these real coordinates correspond to
\begin{equation}\label{gt}
g_2(\tau) = (-2\pi i)^2 t_2(\tau), \qquad g_3(\tau)=(-2\pi i)^3t_3(\tau),
\end{equation}
where $t_i(\tau)$ are the solution (\ref{Ramanujansolution}) of the Ramanujan equations (\ref{Ramanujaneq}). By $G$ isomorphism $(0,g_2(\tau),g_3(\tau))\equiv (t_1(\tau),t_2(\tau),t_3(\tau))$ in  $\mathtt{T_R}$, where an additive transformation of $G$ sets the first coordinate to zero. The values of $t_i(\tau)$ in (\ref{Ramanujansolution}), for imaginary $\tau$ implies 
by  the algebraic morphism $f$ in \ref{TRTH} that the solution $(T_1(\tau),T_2(\tau),T_3(\tau))$ of the Darboux-Halphen equations (\ref{halphen}) is  restricted to pure imaginary values when $\tau$ is imaginary, as mentioned before.

Also by  the morphism $f$ in  and group isomorphisms in $\mathtt{T_R}$ and $\mathtt{T_H}$, $\mathtt{T_\Omega}$ maps to the section $ \{ (t_1,t_2,t_3)\in \mathtt{T_R}| (-2\pi i)^{i}t_i\in \mathbb{R}\}$ of the moduli $\mathtt{T_R}$  of the enhanced elliptic curve $E_t$ (\ref{Et}). This section of $\mathtt{T_R}$ satisfies the  reality condition of the spectral curve, and we call it the real section of $\mathtt{T_R}$ .

Therefore we recognize that $\mathtt{T_\Omega}$ is a six-fold cover of the moduli of the enhanced spectral curve, which we define below:
\begin{definition}
The moduli of an enhanced spectral curve $\tilde{S}_k$ of a k-monopole is the real section of the moduli of   the enhanced algebraic curve $(S_k^{\C}, \{[\alpha_i]\})$ where $S_k^{\C}$ is the family of algebraic curves in $\stlines$  given by (\ref{Spolinomial}) without the real structure constraints and $\{[\alpha_i]\}$ is a basis of classes of algebraic de Rham cohomology  $H^1_\dR(S_k^{\C})$ of  differential 1-forms on $S_k^{\C}$ with fixed intersection matrix $\Phi_{ij}=\langle \alpha_i,\alpha_j\rangle$.  
\end{definition}
%The torsion data, when present, is necessary to sort one enhanced curve among a finite number of them when a congruence subgroup $\Gamma \in SL_2(\mathbb{Z})$ is used to define the algebraic structure of the moduli.
Below we summarize our findings in a form of a theorem:
\begin{theorem}
The moduli of enhanced spectral curves  of $SU(2)$ monopoles of charge 2 quotient by $SO(3)$ action and  translations in $\mathbb{R}^3$ corresponds to $\mathtt{T_\Omega}$, a real section of $\mathtt{T_H}$, quotient by permutations of $\Omega_1, \Omega_2$ and $\Omega_3$. Furthermore, the self-dual curvature equation (\ref{sdeq}) corresponds to the Ramanujan vector field in $\mathtt{T_R}$ upon reparametrization  (\ref{reparam},\ref{tauomega}) and $f$ isomorphism (\ref{morphismHR}). 
\end{theorem}
%\vspace{0.5cm}
\begin{proof}

This theorem is  proved by $f$ isomorphism under restriction  $|T_1|< |T_2| <|T_3|$ in $\mathtt{T_H}$. % The spectral curve (\ref{g2g3}) is enhanced to $E_t$ (\ref{Et}), with real parameters $(t_1,t_2,t_3)$ and
The SO(3) isometry in $M^0_2$ means that this 4-manifold can be expressed by 3-dimensional SO(3) orbits and a orthogonal trajectory parametrized by $\rho$. Hence, $M^0_2$ quotient by SO(3) action is a space of one real dimension parametrized by $\rho$. Accordingly, the spectral curve $S_2$, quotient by SO(3) action and translations, depends on a single parameter, after genus or ES\cite{Ercolani:1989tp}  constraint is imposed. When we work with the moduli of enhanced spectral curves, two extra parameters related to Hodge structure of the curve are added, but the vector field equation Gauss-Manin connection in disguise shows that   
the three parameters of the moduli $\mathtt{T_{H}}$ depend on a single real parameter. 
\end{proof}
 
The moduli of $\tilde{S}_2$ does not include the point $r_1=0$ and $r_2=\pi^2/4$ of zero discriminant  where the curve $S_2$ degenerates to two $S_1$. This point can be mapped to a point of zero discriminant in $E_T$ (\ref{ET}) and it is given by the $\tau=i \infty$ limit in the solution of the system (\ref{halphereal}), with appropriate lattice scaling to match $g_2$ in (\ref{eisenstein},\ref{r2}):
\begin{align}
\Omega_i(\rho)=\frac{\pi}{r_1}\tfrac{\partial}{\partial\rho}\left(\log \theta_{i+1}(i\rho)\right),\quad \textnormal{where}\quad \rho = -i \tau\quad \textnormal{and}\\
\left \{ \begin{array}{l} \theta_2(\tau ):=\sum_{n=-\infty}^\infty
q^{\frac{1}{2}(n+\frac{1}{2})^2}
\\
\theta_3(\tau ):=\sum_{n=-\infty}^\infty q^{\frac{1}{2}n^2}
\\
\theta_4(\tau ):=\sum_{n=-\infty}^\infty (-1)^nq^{\frac{1}{2}n^2}
\end{array} \right.,
\ q=e^{2\pi i \tau },\  \tau \in \uhp.
\end{align}
Notice that $r_1/r_2$ is function of $\tau=i\rho$ and in the limit $\rho\rightarrow\infty$ we can write a explicit relation 
\begin{equation}
r_2\rightarrow \pi^2/4,\quad \textnormal{and} \quad r_1\rightarrow \pi^2q^{1/4}
\end{equation}
Therefore, near the limit $\rho\rightarrow\infty\; (q<< 1)$ we have
\begin{equation}
\Omega_1\approx -\frac{q^{1/4}}{4},\quad \Omega_2\approx -2q^{1/4},\quad \Omega_3\approx 2q^{1/4}
\end{equation}
and the metric of $M^0_2$ becomes
\begin{equation}
ds^2\approx 4q^{1/4}d\rho^2+q^{3/4}(\sigma_1)^2+\dfrac{q^{-1/4}}{4}\left((\sigma_3)^2+(\sigma_2)^2\right)
\end{equation}
 The metric is singular at $q=0$, but asymptotically, two of the coefficients of the metric (\ref{bianchiix}) become equal. In this case, the isometry grows to $SO(3)\times SO(2)$ \cite{manin2015symbolic}, where $SO(2)$ action corresponds to the axial symmetry of two 1-monopole solution and it corresponds to the $S^1$ isotropy  subgroup of the spectral curve \cite{Hurtubise:1983hl}. In other words, the asymptotic behavior of the metric confirms the behavior of the spectral curve at $r_1=0$.
 
 Furthermore, at infinite $\rho$ distance, the $SO(3)\times SO(2)$ orbit is a 2-torus Hopf fibration of the 3-sphere $S^3$, which confirms the fact that the manifold $M^0_2$ is an asymptotically locally Euclidean (ALE) space. This fact together with self-duality equations, characterizes it as a gravitational instanton configuration\cite{MR535151}. These are  elements to take in consideration for finding metrics of  $ M^0_k, k>2$.

\section{Conclusion}\label{conclusion}
We hope this article pave the way for future contributions in understanding the moduli $M^0_k$ of SU(2) monopoles in YMH theory. Among the obstacles, there are the growing computational challenge of Gauss-Manin Connection in Disguise for larger $k$ and the need to understand the homomorphism between vector fields in the enhanced spectral curves and curvature equations in the moduli $M^0_k$.  

The later obstacle is related to the fact that the universal families of curves can be written using different choices of parametrization, which yield different set of differential equations with  different algebraic group of transformations of the moduli (where lattice scaling is one of the operations) \cite{ho14}. In the well known case of elliptic curves, the different choices of parametrization of the universal families for the enhanced elliptic curves takes place according to the choices of congruence subgroups $\Gamma$ of the modular group $SL_2(\mathbb{Z})$\cite{ho14}. The moduli parametrization of the enhanced curves are lifted to modular-type functions under algebraic  group action in the moduli with distinct weights.  We notice that the canonical form (\ref{Spolinomial}) of spectral curves of $k$-monopoles leads to Ramanujan type of parametrization with parameters with distinct scaling weights, while we expect that the curvature equations from $M^0_k$ leads to Darboux-Halphen type of parametrization with parameters with same (scaling) weight. Therefore, another future step in this projects is  to find new modular-type functions attached to $S_k$ curves that will play a role on defining the metric of $M^0_k$. Such role will depend on the symmetries of the moduli that may define the behavior of the metric of $M^0_k$ under algebraic group of tranformations of the moduli $\mathtt{T}$ of the enhanced curve.In summary, the results of $M^0_2$ suggest that the terms of the metric of $M^0_k$ will be  (quasi-)homogeneous polynomials or rational functions of modular-type functions that will correspond to coordinates of the moduli $\mathtt{T}$ of enhanced spectral curves $\tilde{S}_k$ satisfying a unique set of  vector field equations in $\mathtt{T}$ \cite{movasati2015gauss}.

\begin{acknowledgments}
During the period of preparation of the manuscript MACT was fully sponsored by CNpQ-Brasil. The  author profited by the rich academic environment at IMPA and by many interactions with Hossein Movasati, whose work is the basis of this project. My sincere thanks go to him and my colleagues from the project on DH equations R. Roychowdhury, Y. Nikdelan and J. A. Cruz Morales, with whom my first thoughts on this project were shared. 
\end{acknowledgments}
% If in two-column mode, this environment will change to single-column format so that long equations can be displayed. 
% Use only when necessary.
%\begin{widetext}
%$$\mbox{put long equation here}$$
%\end{widetext}

% Figures should be put into the text as floats. 
% Use the graphics or graphicx packages (distributed with LaTeX2e).
% See the LaTeX Graphics Companion by Michel Goosens, Sebastian Rahtz, and Frank Mittelbach for examples. 
%
% Here is an example of the general form of a figure:
% Fill in the caption in the braces of the \caption{} command. 
% Put the label that you will use with \ref{} command in the braces of the \label{} command.
%
% \begin{figure}
% \includegraphics{}%
% \caption{\label{}}%
% \end{figure}

% Tables may be be put in the text as floats.
% Here is an example of the general form of a table:
% Fill in the caption in the braces of the \caption{} command. Put the label
% that you will use with \ref{} command in the braces of the \label{} command.
% Insert the column specifiers (l, r, c, d, etc.) in the empty braces of the                   
% \begin{tabular}{} command.
%
% \begin{table}
% \caption{\label{} }
% \begin{tabular}{}
% \end{tabular}
% \end{table}

% If you have acknowledgments, this puts in the proper section head.
%\begin{acknowledgments}
% Put your acknowledgments here.
%\end{acknowledgments}

% Create the reference section using BibTeX:
\bibliographystyle{aipnum4-1}%{model1-num-names} 
\bibliography{SpectralGMCD}

%merlin.mbs aipnum4-1.bst 2010-07-25 4.21a (PWD, AO, DPC) hacked
%Control: key (0)
%Control: author (8) initials jnrlst
%Control: editor formatted (1) identically to author
%Control: production of article title (-1) disabled
%Control: page (0) single
%Control: year (1) truncated
%Control: production of eprint (0) enabled
\begin{thebibliography}{25}%
\makeatletter
\providecommand \@ifxundefined [1]{%
 \@ifx{#1\undefined}
}%
\providecommand \@ifnum [1]{%
 \ifnum #1\expandafter \@firstoftwo
 \else \expandafter \@secondoftwo
 \fi
}%
\providecommand \@ifx [1]{%
 \ifx #1\expandafter \@firstoftwo
 \else \expandafter \@secondoftwo
 \fi
}%
\providecommand \natexlab [1]{#1}%
\providecommand \enquote  [1]{``#1''}%
\providecommand \bibnamefont  [1]{#1}%
\providecommand \bibfnamefont [1]{#1}%
\providecommand \citenamefont [1]{#1}%
\providecommand \href@noop [0]{\@secondoftwo}%
\providecommand \href [0]{\begingroup \@sanitize@url \@href}%
\providecommand \@href[1]{\@@startlink{#1}\@@href}%
\providecommand \@@href[1]{\endgroup#1\@@endlink}%
\providecommand \@sanitize@url [0]{\catcode `\\12\catcode `\$12\catcode
  `\&12\catcode `\#12\catcode `\^12\catcode `\_12\catcode `\%12\relax}%
\providecommand \@@startlink[1]{}%
\providecommand \@@endlink[0]{}%
\providecommand \url  [0]{\begingroup\@sanitize@url \@url }%
\providecommand \@url [1]{\endgroup\@href {#1}{\urlprefix }}%
\providecommand \urlprefix  [0]{URL }%
\providecommand \Eprint [0]{\href }%
\providecommand \doibase [0]{http://dx.doi.org/}%
\providecommand \selectlanguage [0]{\@gobble}%
\providecommand \bibinfo  [0]{\@secondoftwo}%
\providecommand \bibfield  [0]{\@secondoftwo}%
\providecommand \translation [1]{[#1]}%
\providecommand \BibitemOpen [0]{}%
\providecommand \bibitemStop [0]{}%
\providecommand \bibitemNoStop [0]{.\EOS\space}%
\providecommand \EOS [0]{\spacefactor3000\relax}%
\providecommand \BibitemShut  [1]{\csname bibitem#1\endcsname}%
\let\auto@bib@innerbib\@empty
%</preamble>
\bibitem [{\citenamefont {Hitchin}(1982)}]{1982CMaPh..83..579H}%
  \BibitemOpen
  \bibfield  {author} {\bibinfo {author} {\bibfnamefont {N.~J.}\ \bibnamefont
  {Hitchin}},\ }\href {http://adsabs.harvard.edu/abs/1982CMaPh..83..579H}
  {\bibfield  {journal} {\bibinfo  {journal} {Communications in Mathematical
  Physics}\ }\textbf {\bibinfo {volume} {83}},\ \bibinfo {pages} {579}
  (\bibinfo {year} {1982})}\BibitemShut {NoStop}%
\bibitem [{\citenamefont {Morales}\ \emph {et~al.}(2018)\citenamefont
  {Morales}, \citenamefont {Movasati}, \citenamefont {Nikdelan}, \citenamefont
  {Roychowdhury},\ and\ \citenamefont {Torres}}]{Morales:2017phm}%
  \BibitemOpen
  \bibfield  {author} {\bibinfo {author} {\bibfnamefont {J.~A.~C.}\
  \bibnamefont {Morales}}, \bibinfo {author} {\bibfnamefont {H.}~\bibnamefont
  {Movasati}}, \bibinfo {author} {\bibfnamefont {Y.}~\bibnamefont {Nikdelan}},
  \bibinfo {author} {\bibfnamefont {R.}~\bibnamefont {Roychowdhury}}, \ and\
  \bibinfo {author} {\bibfnamefont {M.~A.~C.}\ \bibnamefont {Torres}},\ }\href
  {\doibase 10.3842/SIGMA.2018.003} {\bibfield  {journal} {\bibinfo  {journal}
  {SIGMA}\ }\textbf {\bibinfo {volume} {14}},\ \bibinfo {pages} {003} (\bibinfo
  {year} {2018})},\ \Eprint {http://arxiv.org/abs/1709.09682} {arXiv:1709.09682
  [math.DG]} \BibitemShut {NoStop}%
%%CITATION = ARXIV:1709.09682;%%
\bibitem [{\citenamefont {Hurtubise}(1983)}]{Hurtubise:1983hl}%
  \BibitemOpen
  \bibfield  {author} {\bibinfo {author} {\bibfnamefont {J.}~\bibnamefont
  {Hurtubise}},\ }\href {\doibase 10.1007/BF01210845} {\bibfield  {journal}
  {\bibinfo  {journal} {Communications in Mathematical Physics}\ }\textbf
  {\bibinfo {volume} {92}},\ \bibinfo {pages} {195} (\bibinfo {year}
  {1983})}\BibitemShut {NoStop}%
\bibitem [{\citenamefont {Alim}\ \emph {et~al.}(2016)\citenamefont {Alim},
  \citenamefont {Movasati}, \citenamefont {Scheidegger},\ and\ \citenamefont
  {Yau}}]{alim2016gauss}%
  \BibitemOpen
  \bibfield  {author} {\bibinfo {author} {\bibfnamefont {M.}~\bibnamefont
  {Alim}}, \bibinfo {author} {\bibfnamefont {H.}~\bibnamefont {Movasati}},
  \bibinfo {author} {\bibfnamefont {E.}~\bibnamefont {Scheidegger}}, \ and\
  \bibinfo {author} {\bibfnamefont {S.-T.}\ \bibnamefont {Yau}},\ }\href@noop
  {} {\bibfield  {journal} {\bibinfo  {journal} {Communications in Mathematical
  Physics}\ }\textbf {\bibinfo {volume} {344}},\ \bibinfo {pages} {889}
  (\bibinfo {year} {2016})}\BibitemShut {NoStop}%
\bibitem [{\citenamefont {Movasati}(2011)}]{ho06-1}%
  \BibitemOpen
  \bibfield  {author} {\bibinfo {author} {\bibfnamefont {H.}~\bibnamefont
  {Movasati}},\ }\href@noop {} {\emph {\bibinfo {title} {Multiple Integrals and
  Modular Differential Equations}}},\ 28th Brazilian Mathematics Colloquium\
  (\bibinfo  {publisher} {Instituto de Matem{\'a}tica Pura e Aplicada, IMPA},\
  \bibinfo {year} {2011})\BibitemShut {NoStop}%
\bibitem [{\citenamefont {Movasati}(2012{\natexlab{a}})}]{ho14}%
  \BibitemOpen
  \bibfield  {author} {\bibinfo {author} {\bibfnamefont {H.}~\bibnamefont
  {Movasati}},\ }\href {http://ambp.cedram.org/item?id={AMBP_2012__19_2_307_0}}
  {\bibfield  {journal} {\bibinfo  {journal} {Ann. Math. Blaise Pascal}\
  }\textbf {\bibinfo {volume} {19}},\ \bibinfo {pages} {307} (\bibinfo {year}
  {2012}{\natexlab{a}})}\BibitemShut {NoStop}%
\bibitem [{\citenamefont {Atiyah}\ and\ \citenamefont
  {Ward}(1977)}]{atiyah1977instantons}%
  \BibitemOpen
  \bibfield  {author} {\bibinfo {author} {\bibfnamefont {M.~F.}\ \bibnamefont
  {Atiyah}}\ and\ \bibinfo {author} {\bibfnamefont {R.~S.}\ \bibnamefont
  {Ward}},\ }\href@noop {} {\bibfield  {journal} {\bibinfo  {journal}
  {Communications in Mathematical Physics}\ }\textbf {\bibinfo {volume} {55}},\
  \bibinfo {pages} {117} (\bibinfo {year} {1977})}\BibitemShut {NoStop}%
\bibitem [{\citenamefont {Atiyah}\ and\ \citenamefont
  {Hitchin}(2014)}]{atiyah2014geometry}%
  \BibitemOpen
  \bibfield  {author} {\bibinfo {author} {\bibfnamefont {M.~F.}\ \bibnamefont
  {Atiyah}}\ and\ \bibinfo {author} {\bibfnamefont {N.}~\bibnamefont
  {Hitchin}},\ }\href@noop {} {\emph {\bibinfo {title} {The geometry and
  dynamics of magnetic monopoles}}}\ (\bibinfo  {publisher} {Princeton
  University Press},\ \bibinfo {year} {2014})\BibitemShut {NoStop}%
\bibitem [{\citenamefont {Ward}\ and\ \citenamefont
  {Wells}(1990)}]{ward1990twistor}%
  \BibitemOpen
  \bibfield  {author} {\bibinfo {author} {\bibfnamefont {R.~S.}\ \bibnamefont
  {Ward}}\ and\ \bibinfo {author} {\bibfnamefont {R.~O.}\ \bibnamefont
  {Wells}},\ }\href@noop {} {\emph {\bibinfo {title} {Twistor geometry and
  field theory}}},\ Vol.~\bibinfo {volume} {4}\ (\bibinfo  {publisher}
  {Cambridge University Press Cambridge},\ \bibinfo {year} {1990})\BibitemShut
  {NoStop}%
\bibitem [{\citenamefont {Bogomol’nyi}(1976)}]{Bogomol`nyi_1976}%
  \BibitemOpen
  \bibfield  {author} {\bibinfo {author} {\bibfnamefont {E.}~\bibnamefont
  {Bogomol’nyi}},\ }\href@noop {} {\bibfield  {journal} {\bibinfo  {journal}
  {Sov. J. Nucl. Phys. (Engl. Transl.); (United States)}\ }\textbf {\bibinfo
  {volume} {24:4}} (\bibinfo {year} {1976})}\BibitemShut {NoStop}%
\bibitem [{\citenamefont {Hitchin}(1983)}]{hitchin1983construction}%
  \BibitemOpen
  \bibfield  {author} {\bibinfo {author} {\bibfnamefont {N.~J.}\ \bibnamefont
  {Hitchin}},\ }\href@noop {} {\bibfield  {journal} {\bibinfo  {journal}
  {Communications in Mathematical Physics}\ }\textbf {\bibinfo {volume} {89}},\
  \bibinfo {pages} {145} (\bibinfo {year} {1983})}\BibitemShut {NoStop}%
\bibitem [{\citenamefont {Ercolani}\ and\ \citenamefont
  {Sinha}(1989)}]{Ercolani:1989tp}%
  \BibitemOpen
  \bibfield  {author} {\bibinfo {author} {\bibfnamefont {N.}~\bibnamefont
  {Ercolani}}\ and\ \bibinfo {author} {\bibfnamefont {A.}~\bibnamefont
  {Sinha}},\ }\href {\doibase 10.1007/BF01218409} {\bibfield  {journal}
  {\bibinfo  {journal} {Commun. Math. Phys.}\ }\textbf {\bibinfo {volume}
  {125}},\ \bibinfo {pages} {385} (\bibinfo {year} {1989})}\BibitemShut
  {NoStop}%
%%CITATION = CMPHA,125,385;%%
\bibitem [{\citenamefont {Ward}(1981)}]{Ward:1981jb}%
  \BibitemOpen
  \bibfield  {author} {\bibinfo {author} {\bibfnamefont {R.~S.}\ \bibnamefont
  {Ward}},\ }\href {\doibase 10.1007/BF01208497} {\bibfield  {journal}
  {\bibinfo  {journal} {Commun. Math. Phys.}\ }\textbf {\bibinfo {volume}
  {79}},\ \bibinfo {pages} {317} (\bibinfo {year} {1981})}\BibitemShut
  {NoStop}%
%%CITATION = CMPHA,79,317;%%
\bibitem [{\citenamefont {Sutcliffe}(1997)}]{Sutcliffe:1997ec}%
  \BibitemOpen
  \bibfield  {author} {\bibinfo {author} {\bibfnamefont {P.~M.}\ \bibnamefont
  {Sutcliffe}},\ }\href {\doibase 10.1142/S0217751X97002504} {\bibfield
  {journal} {\bibinfo  {journal} {Int. J. Mod. Phys.}\ }\textbf {\bibinfo
  {volume} {A12}},\ \bibinfo {pages} {4663} (\bibinfo {year} {1997})},\ \Eprint
  {http://arxiv.org/abs/hep-th/9707009} {arXiv:hep-th/9707009 [hep-th]}
  \BibitemShut {NoStop}%
%%CITATION = HEP-TH/9707009;%%
\bibitem [{\citenamefont {Movasati}(2008)}]{movasati2008differential}%
  \BibitemOpen
  \bibfield  {author} {\bibinfo {author} {\bibfnamefont {H.}~\bibnamefont
  {Movasati}},\ }\href@noop {} {\bibfield  {journal} {\bibinfo  {journal} {The
  Ramanujan Journal}\ }\textbf {\bibinfo {volume} {17}},\ \bibinfo {pages} {53}
  (\bibinfo {year} {2008})}\BibitemShut {NoStop}%
\bibitem [{\citenamefont
  {Movasati}(2012{\natexlab{b}})}]{movasati2012ramanujan}%
  \BibitemOpen
  \bibfield  {author} {\bibinfo {author} {\bibfnamefont {H.}~\bibnamefont
  {Movasati}},\ }\href@noop {} {\bibfield  {journal} {\bibinfo  {journal}
  {manuscripta mathematica}\ }\textbf {\bibinfo {volume} {139}},\ \bibinfo
  {pages} {495} (\bibinfo {year} {2012}{\natexlab{b}})}\BibitemShut {NoStop}%
\bibitem [{\citenamefont
  {Griffiths}(1968{\natexlab{a}})}]{griffiths1968periods1}%
  \BibitemOpen
  \bibfield  {author} {\bibinfo {author} {\bibfnamefont {P.~A.}\ \bibnamefont
  {Griffiths}},\ }\href@noop {} {\bibfield  {journal} {\bibinfo  {journal}
  {American Journal of Mathematics}\ }\textbf {\bibinfo {volume} {90}},\
  \bibinfo {pages} {568} (\bibinfo {year} {1968}{\natexlab{a}})}\BibitemShut
  {NoStop}%
\bibitem [{\citenamefont
  {Griffiths}(1968{\natexlab{b}})}]{griffiths1968periods2}%
  \BibitemOpen
  \bibfield  {author} {\bibinfo {author} {\bibfnamefont {P.~A.}\ \bibnamefont
  {Griffiths}},\ }\href@noop {} {\bibfield  {journal} {\bibinfo  {journal}
  {American Journal of Mathematics}\ }\textbf {\bibinfo {volume} {90}},\
  \bibinfo {pages} {805} (\bibinfo {year} {1968}{\natexlab{b}})}\BibitemShut
  {NoStop}%
\bibitem [{\citenamefont {Movasati}(2014)}]{movasati2014gauss}%
  \BibitemOpen
  \bibfield  {author} {\bibinfo {author} {\bibfnamefont {H.}~\bibnamefont
  {Movasati}},\ }\href@noop {} {\bibfield  {journal} {\bibinfo  {journal}
  {arXiv preprint arXiv:1411.1766}\ } (\bibinfo {year} {2014})}\BibitemShut
  {NoStop}%
\bibitem [{\citenamefont {Movasati}(2015)}]{movasati2015gauss}%
  \BibitemOpen
  \bibfield  {author} {\bibinfo {author} {\bibfnamefont {H.}~\bibnamefont
  {Movasati}},\ }\href@noop {} {\bibfield  {journal} {\bibinfo  {journal}
  {Surveys of Modern Mathematics, IP, Boston. Available online at http://w3.
  impa. br/hossein/myarticles/GMCD-MQCY3. pdf}\ }\textbf {\bibinfo {volume}
  {170}} (\bibinfo {year} {2015})}\BibitemShut {NoStop}%
\bibitem [{\citenamefont {Movasati}\ and\ \citenamefont
  {Nikdelan}(2016)}]{movasati2016gauss}%
  \BibitemOpen
  \bibfield  {author} {\bibinfo {author} {\bibfnamefont {H.}~\bibnamefont
  {Movasati}}\ and\ \bibinfo {author} {\bibfnamefont {Y.}~\bibnamefont
  {Nikdelan}},\ }\href@noop {} {\bibfield  {journal} {\bibinfo  {journal}
  {arXiv preprint arXiv:1603.09411}\ } (\bibinfo {year} {2016})}\BibitemShut
  {NoStop}%
\bibitem [{\citenamefont {Ramanujan}(1916)}]{ra16}%
  \BibitemOpen
  \bibfield  {author} {\bibinfo {author} {\bibfnamefont {S.}~\bibnamefont
  {Ramanujan}},\ }\href@noop {} {\bibfield  {journal} {\bibinfo  {journal}
  {Trans. Cambridge Philos. Soc.}\ }\textbf {\bibinfo {volume} {22}},\ \bibinfo
  {pages} {159} (\bibinfo {year} {1916})}\BibitemShut {NoStop}%
\bibitem [{\citenamefont {Bianchi}(2001)}]{Bianchi2001}%
  \BibitemOpen
  \bibfield  {author} {\bibinfo {author} {\bibfnamefont {L.}~\bibnamefont
  {Bianchi}},\ }\href {\doibase 10.1023/A:1015357132699} {\bibfield  {journal}
  {\bibinfo  {journal} {General Relativity and Gravitation}\ }\textbf {\bibinfo
  {volume} {33}},\ \bibinfo {pages} {2171} (\bibinfo {year}
  {2001})}\BibitemShut {NoStop}%
\bibitem [{\citenamefont {Manin}\ and\ \citenamefont
  {Marcolli}(2015)}]{manin2015symbolic}%
  \BibitemOpen
  \bibfield  {author} {\bibinfo {author} {\bibfnamefont {Y.}~\bibnamefont
  {Manin}}\ and\ \bibinfo {author} {\bibfnamefont {M.}~\bibnamefont
  {Marcolli}},\ }\href@noop {} {\bibfield  {journal} {\bibinfo  {journal}
  {arXiv preprint arXiv:1504.04005}\ } (\bibinfo {year} {2015})}\BibitemShut
  {NoStop}%
\bibitem [{\citenamefont {Gibbons}\ and\ \citenamefont
  {Pope}(1979)}]{MR535151}%
  \BibitemOpen
  \bibfield  {author} {\bibinfo {author} {\bibfnamefont {G.~W.}\ \bibnamefont
  {Gibbons}}\ and\ \bibinfo {author} {\bibfnamefont {C.~N.}\ \bibnamefont
  {Pope}},\ }\href {http://projecteuclid.org/euclid.cmp/1103905050} {\bibfield
  {journal} {\bibinfo  {journal} {Comm. Math. Phys.}\ }\textbf {\bibinfo
  {volume} {66}},\ \bibinfo {pages} {267} (\bibinfo {year} {1979})}\BibitemShut
  {NoStop}%
\end{thebibliography}%

\end{document}